\documentclass{article}

\usepackage{arxiv}

\usepackage[utf8]{inputenc} % allow utf-8 input
\usepackage[T1]{fontenc}    % use 8-bit T1 fonts
\usepackage{hyperref}       % hyperlinks
\usepackage{url}            % simple URL typesetting
\usepackage{booktabs}       % professional-quality tables
\usepackage{amsfonts}       % blackboard math symbols
\usepackage{nicefrac}       % compact symbols for 1/2, etc.
\usepackage{microtype}      % microtypography
\usepackage{lipsum}
\usepackage{graphicx}
\graphicspath{ {./images/} }

\usepackage{times}
\usepackage{soul}
\usepackage{url}
\usepackage[small]{caption}
\usepackage{graphicx}
\usepackage{amsmath}
\usepackage{amsthm}
\usepackage{booktabs}
\usepackage{pifont}
\usepackage[linesnumbered]{algorithm2e}

\newtheorem{example}{Example}

\newtheorem{definition}{Definition}
\newtheorem{proposition}{Proposition}

\usepackage[misc]{ifsym} % Required for the \Letter symbol

\usepackage{color}
\usepackage[table]{xcolor}
\usepackage{url}

\definecolor{lightred}{RGB}{255,230,230}
\definecolor{lightgreen}{RGB}{230,255,230}
\definecolor{darkgreen}{RGB}{200,255,200}
\definecolor{lightgrey}{RGB}{245,245,245}

\newcommand{\cmark}{\ding{51}}%
\newcommand{\xmark}{\ding{55}}%

\newcommand{\yesmark}{\cellcolor{lightgreen}\cmark}
\newcommand{\nomark}{\cellcolor{lightred}\xmark}

\usepackage{amssymb}
\newcommand{\twod}{\ensuremath{\mathbb{B}}}
\newcommand{\threed}{\ensuremath{\mathbb{B}_{\star}}}

\newcommand{\acceptcon}{\ensuremath{\varphi}}

\newcommand{\tval}{\ensuremath{1}}
\newcommand{\fval}{\ensuremath{0}}
\newcommand{\uval}{\ensuremath{\star}}

\newcommand{\inorder}{\ensuremath{\leq_{i}}}
\newcommand{\inorderst}{\ensuremath{<_{i}}}

\newcommand{\charoper}{\ensuremath{\Gamma_{D}}}

\newcolumntype{L}[1]{>{\raggedright\let\newline\\\arraybackslash\hspace{0pt}}m{#1}}
\newcolumntype{C}[1]{>{\centering\let\newline\\\arraybackslash\hspace{0pt}}m{#1}}
\newcolumntype{R}[1]{>{\raggedleft\let\newline\\\arraybackslash\hspace{0pt}}m{#1}}

\newcommand{\smallcell}[1]{{\scriptsize #1}}

\usepackage{xspace}
\newcommand{\toolname}{\texttt{BAss}\xspace}
\newcommand{\adfbdd}{\texttt{adf-bdd}\xspace}
\newcommand{\kadf}{\texttt{k++adf}\xspace}
\newcommand{\tsconj}{\texttt{ts-conj}\xspace}
\newcommand{\fasp}{\texttt{fASP}\xspace}
\newcommand{\biolqm}{\texttt{bioLQM}\xspace}
\newcommand{\godiamond}{\texttt{goDiamond}\xspace}
\newcommand{\yadf}{\texttt{yadf}\xspace}
\newcommand{\mpbn}{\texttt{mpbn}\xspace}
\newcommand{\pybn}{\texttt{pyboolnet}\xspace}
\newcommand{\aeon}{\texttt{aeon}\xspace}
\newcommand{\ruddy}{\texttt{ruddy}\xspace}
\newcommand{\he}{\texttt{H.E.}\xspace}

\newcommand{\ifftext}{\text{iff}\xspace}
\newcommand{\wrt}{\text{w.r.t.}\xspace}

\newcommand{\fulsym}{\text{fully-symbolic}\xspace}

\newcommand{\threeint}{\text{3-valued interpretation}\xspace}
\newcommand{\twoint}{\text{2-valued interpretation}\xspace}
\newcommand{\threeints}{\text{3-valued interpretations}\xspace}
\newcommand{\twoints}{\text{2-valued interpretations}\xspace}

\newcommand{\grint}{\text{grounded interpretation}\xspace}

\newcommand{\twomodel}{\text{2-valued model}\xspace}
\newcommand{\stbmodel}{\text{stable model}\xspace}

\newcommand{\adints}{\text{admissible interpretations}\xspace}
\newcommand{\coints}{\text{complete interpretations}\xspace}

\newcommand{\prints}{\text{preferred interpretations}\xspace}
\newcommand{\twomodels}{\text{2-valued models}\xspace}
\newcommand{\stbmodels}{\text{stable models}\xspace}

\newcommand{\adset}[1]{\texttt{ad}(#1)\xspace}
\newcommand{\grset}[1]{\texttt{gr}(#1)\xspace}
\newcommand{\coset}[1]{\texttt{co}(#1)\xspace}
\newcommand{\prset}[1]{\texttt{pr}(#1)\xspace}
\newcommand{\twomset}[1]{\texttt{2m}(#1)\xspace}
\newcommand{\stbset}[1]{\texttt{stb}(#1)\xspace}

\title{BAss: Symbolic Reasoning in Abstract Dialectical Frameworks}

\author{%
	Samuel Pastva \\
	Faculty of Informatics, Masaryk University, Brno, Czechia \\
	\texttt{xpastva@fi.muni.cz} \\
	\And
	Van-Giang Trinh \\
	Faculty of Computer Science and Engineering, \\
	Ho Chi Minh City University of Technology (HCMUT), VNU-HCM, \\
	Ho Chi Minh City, Vietnam\\
	\texttt{van-giang.trinh@hcmut.edu.vn}
}

\begin{document}
	
	\maketitle
	
	\begin{abstract}
		We present \toolname{} (BDD-based ADF symbolic solver), a novel analysis tool for \textit{Abstract Dialectical Frameworks} (ADFs) based on \textit{Binary Decision Diagrams} (BDDs).
		It supports the \fulsym computation of all \textit{admissible}, \textit{complete}, and \textit{preferred interpretations}, as well as \textit{2-valued} and \textit{stable models} of an ADF.
		Our approach is inspired by the recently discovered equivalence between \textit{Boolean Networks} (BNs) and ADFs by Heyninck et al. (2024) and Azpeitia et al. (2024), significantly extending the capabilities of current BDD-based tools \biolqm, \aeon, and \adfbdd.
		We conducted experiments on a large-scale collection of real-world models from both the BN and ADF communities.
		Our results show that \toolname{} dramatically outperforms previous BDD-based tools and is competitive with (even significantly better in some cases) state-of-the-art SAT/ASP-based methods, particularly in scenarios involving large solution spaces.
		Notably, \toolname{} is able to enumerate all fixed points or minimal trap spaces of certain biological networks beyond the reach of existing tools, thereby enabling new analyses and case studies in systems biology.
		These results highlight the practical relevance of symbolic reasoning for complex real-world applications, particularly in systems biology and formal argumentation.
	\end{abstract}
	
	\section{Introduction}
	\label{sec:introduction}
	
	%%% AFs
	Argumentation has emerged as a central area in AI, offering formal models for reasoning with conflicting information and justifications~\cite{BCD2007,ABGHPRST2017}.
	A foundational approach in this field is Dung's \textit{Abstract Argumentation Frameworks} (AFs)~\cite{Dung1995}, which represent arguments as abstract entities and define their interactions via a binary attack relation. 
	Despite their simplicity, AFs capture a wide range of nonmonotonic reasoning patterns and have inspired extensive research on semantics, algorithms, and complexity~\cite{CDGWW2015}.
	Their computational appeal is reflected in the development of numerous solvers and the organization of the \textit{International Competition on Computational Models of Argumentation} (ICCMA)~\cite{TV2017}, where SAT-based approaches have shown strong empirical performance~\cite{NJ2020}.
	
	While AFs have been successful, their conceptual simplicity---as directed graphs of abstract arguments and attacks---can be too restrictive for modeling complex dependencies~\cite{Strass2013}.
	\emph{Abstract Dialectical Frameworks} (ADFs) address this by associating each argument with a propositional \emph{acceptance condition}~\cite{BW2010,BSEWW2013}, allowing edges to represent arbitrary logical relationships rather than just attacks.
	This enables ADFs to subsume most classical argumentation semantics, including \emph{complete}, \emph{preferred}, and \emph{stable} semantics (the latter inspired by the stable model semantics of logic programs~\cite{GL1988}).
	Interpretations assign arguments one of three truth values (true, false, or unknown), and semantics is defined via a characteristic operator.
	ADFs thus form a highly expressive framework for argumentation, rule-based reasoning, and logic~\cite{BSEWW2013,BESWW2017}, with applications in legal reasoning~\cite{AAB2016}, dialog systems~\cite{Neugebauer2017}, text exploration~\cite{CV2016}, and defeasible theories~\cite{Strass2014}.
	However, reasoning in ADFs is computationally hard: many problems are NP- or coNP-complete, and preferred or stable semantics reaches \(\Sigma^{\text{P}}_2\)-completeness and beyond~\cite{SW2015,LMNWW2022}.
	This motivates ongoing efforts to develop efficient and scalable ADF solvers~\cite{LMNWW2022,EGRW2022}.
	
	%%% ADF-BN equivalence
	\paragraph{ADF-BN equivalence}
	
	Recent work by~\cite{HMJ2024,ASDO2024} has established a formal correspondence between ADFs and \textit{Boolean Networks} (BNs). As a straightforward yet powerful mathematical formalism, BNs have been extensively applied across science and engineering, particularly within systems biology~\cite{Thomas1973,rozum2023boolean}. Formally, a BN is a discrete dynamical system comprising \(n\) Boolean variables and \(n\) corresponding Boolean functions, which dictate discrete state updates according to a specified update scheme. Traditionally, the dynamical behavior of BNs has been analyzed through the lens of \emph{attractors} (minimal inescapable sets of states)~\cite{WSA2012}. More recently, the concept of a \emph{trap space} (an inescapable hypercube of states) was introduced~\cite{HAH2015} and has since emerged as the central focus in BN analysis~\cite{TBPS2024,TKPR2025}.
	
	The aforementioned studies~\cite{HMJ2024,ASDO2024} demonstrate that ADFs inherently function as BNs when their acceptance conditions are translated into Boolean functions. This translation reveals two formal bijections: \twomodels correspond to \emph{fixed points} (singleton attractors), while \prints map to \emph{minimal} trap spaces. Building upon this equivalence, it follows that admissible interpretations of ADFs represent \emph{all} BN trap spaces, and complete interpretations correspond to \emph{fully percolated} trap spaces~\cite{TKPR2025}. 
	
	This theoretical connection effectively bridges two previously disjoint fields---computational argumentation and systems biology---facilitating the reuse of algorithms and software tools~\cite{ASDO2024}. Consequently, many ADF tools can be directly applied to BNs, and vice versa. Given their shared syntax and our primary focus on ADFs, this paper omits formal preliminaries for BNs; readers are referred to~\cite{rozum2023boolean} for comprehensive details.

	\paragraph{Enumeration problems in ADFs and BNs}
	This paper focuses on the enumeration of ADF interpretations or models under various semantics.
	Whereas decision problems have been widely studied in the literature~\cite{LMNWW2022,TBPS2024}, enumeration plays a crucial role in practical applications---such as model analysis, explanation, and verification---where exploring the entire solution space is essential~\cite{EGRW2022,EG2023a,EG2023b,TBS2023CP}.
	Moreover, enumeration is often the computational bottleneck in ADF and BN reasoning tools~\cite{EGRW2022,TBS2023CP,TBPS2024,ASDO2024}, where the number of solutions can be exponential and naive generation strategies quickly become infeasible.
	By targeting enumeration as a first-class objective, we aim to push the boundaries of scalability and completeness in ADF and BN solving.
	
	%%% bass
	Most existing ADF solvers rely on SAT or \textit{Answer Set Programming} (ASP) encodings, which---thanks to highly optimized modern solvers---provide strong performance on many benchmarks.
	Symbolic methods based on \textit{Binary Decision Diagrams} (BDDs) have also been explored, most notably in \adfbdd~\cite{EGRW2022}, showing that BDDs can exploit structural regularities that are often missed by SAT/ASP-based approaches.
	In parallel, the BN community has developed efficient symbolic tools for fixed-point and trap space analysis, such as \biolqm~\cite{Naldi2018}, and for attractor analysis such as \aeon~\cite{BBKPS2020}. 
	These address the solution space explosion problem common in BNs due to the large number of input variables, and can be adapted for use on ADFs.
	
	While BDDs offer a promising foundation for symbolic reasoning in ADFs, enabling both efficient analysis and compact representation~\cite{Strass2018,EGRW2022}, current tools such as \adfbdd lack support for the preferred semantics and face scalability challenges~\cite{EGRW2022}.
	Motivated by the recent ADF–BN equivalence, we introduce \toolname (BDD-based ADF symbolic solver), a new BDD-based solver for ADFs.
	It represents each acceptance condition as a BDD, following the idea of \adfbdd~\cite{EGRW2022}, but significantly extends the core BDD encoding of \biolqm~\cite{Naldi2018} and applies the parametrized principle of the BN attractor analysis tool \aeon~\cite{BBKPS2020} to achieve fully symbolic analysis and improved efficiency.
	
	\paragraph{Contribution}
	\toolname adapts existing symbolic techniques towards the computation of all \emph{admissible}, \emph{complete}, and \emph{preferred} interpretations, as well as 
	\emph{2-valued} and \emph{stable} models.
	Notably, we introduce novel encoding approaches and optimization strategies, which enable \toolname to be the first BDD-based ADF solver to support \fulsym reasoning on the preferred and stable semantics.
	
	We performed an extensive experimental evaluation over 1,200+ benchmark instances from both the BN and ADF domains.
	The results show that \toolname dramatically outperforms current BDD-based methods (\adfbdd, \biolqm, and \aeon) and is competitive with leading SAT/ASP-based methods, such as \kadf, \godiamond, \yadf, \mpbn, \fasp, and \tsconj---often outperforming them on instances that require analysis of large solution sets.
	
	In particular, \toolname also brings concrete benefits to the analysis of biological regulatory networks.
	Thanks to the ADF-BN equivalence, it can enumerate all fixed points or minimal trap spaces of certain Boolean networks that remain unsolvable by existing BN tools.
	This opens the door to new case studies in systems biology (e.g., model-based phenotype discovery), where the complete characterization of dynamical behaviors is often critical.
	
	%%% Paper structure
	The remainder of the paper is organized as follows.
	Section~\ref{sec:related-work} reviews related work on ADF and BN solvers.
	Section~\ref{sec:preliminaries} recalls basic notions of ADFs.
	Section~\ref{sec:analysis-methods} presents the BDD-based methods over different ADF semantics.
	Section~\ref{sec:evaluation} describes our experiments and evaluation of our method.
	It also discusses the biological applicability of our method.
	Section~\ref{sec:conclusion} concludes the paper.
	
	\section{Related Work}
	\label{sec:related-work}
	
	Research on ADF solvers has mainly followed two paradigms: 
	(1) SAT and ASP-based encodings, and (2) symbolic BDD representations. 
	
	%%% SAT/ASP based approaches to ADFs
	Among the SAT-based approaches, \kadf~\cite{LMNWW2022} is one of the most prominent solvers.
	It translates ADF reasoning tasks into propositional satisfiability and supports enumeration under several semantics (admissible, complete, preferred, grounded, stable).
	It possesses efficient SAT encodings, evidenced by its strong empirical performance on models of the ICCMA competitions~\cite{LMNWW2022}.
	Other systems such as \godiamond~\cite{SE2017} and \yadf~\cite{BDHLW2020} use ASP encodings, often excelling on small and medium-sized benchmarks but struggling with scalability~\cite{LMNWW2022,EGRW2022}.
	
	%%% BDD-based approaches to ADFs
	The first BDD-based solver, \adfbdd~\cite{EGRW2022}, represents each acceptance condition as a reduced ordered BDD, enabling the efficient application of the characteristic operator and model evaluation.
	It supports grounded, complete, and stable semantics, and implements fixed-point iteration and model enumeration via symbolic restriction and heuristic-guided search; however, the preferred semantics has not been explored yet.
	While \adfbdd outperforms ASP-based solvers like \yadf~\cite{BDHLW2020} on stable semantics and matches SAT-based tools like \kadf~\cite{LMNWW2022} on grounded semantics, it remains behind the fastest SAT solvers on complete and stable semantics tasks.
	The tool is implemented in Rust, is open-source, and offers both command-line and web interfaces~\cite{EG2023a}.
	In addition, BDDs enable straightforward graphical visualization of the underlying ADFs and their solutions in \adfbdd.
	This showcases the potential of BDDs towards better explainability and debugging of ADFs~\cite{EG2023b}.
	
	%%% Approaches to BNs
	In the BN community, tools such as \biolqm~\cite{Naldi2018} and \pybn~\cite{HAH2015,KSS2017} support the symbolic computation of fixed points and (minimal or maximal) trap spaces using BDDs or ASP encodings.
	\pybn also implements an \textit{Integer Linear Programming} (ILP) encoding, and shows that this encoding is surpassed by the ASP one~\cite{HAH2015}.
	However, the ASP and ILP encodings of \pybn require the computation of prime implicants for each Boolean function, which is a bottleneck since computing even a single prime implicant is NP-hard and the total number of prime implicants can be exponential in the number of function inputs~\cite{CM1978}.
	\biolqm avoids the prime-implicant computation as it characterizes the set of generic trap spaces of a BN using a BDD, then filters this set to get the set of all minimal trap spaces. 
	%By this approach, it requires the computation of all solutions, whereas the methods based on ASP, ILP, or SAT can start enumerating them as they are found.
	However, its main issue is that the number of such generic BN trap spaces may be significantly larger than the number of desirable (e.g., minimal) trap spaces~\cite{Naldi2018,TBS2023}.
	This BDD-based filtering approach has also been implemented in the tool \aeon~\cite{BBHPSS2022} with specific optimizations targeting large solution spaces arising in networks with many input nodes.
	Recently, an efficient SAT-based encoding, called \he, has been proposed for BNs, but it is limited to fixed points~\cite{HSBMBT2025}.
	
	Subsequent ASP encodings~\cite{PKCH2020,VML2024,TBS2023} that were proposed to overcome the prime-implicant bottleneck can still suffer from scalability and efficiency issues. 
	For very large and complex models, these still require the {\em disjunctive normal forms} of all the Boolean functions (and their negations) of the original BN.
	For fixed points, the ASP encoding \fasp~\cite{TBS2023CP} uses a {\em negation normal form} for each Boolean function, which is much more efficient to obtain.
	This encoding was later generalized for minimal trap spaces in \tsconj~\cite{TBPS2024}; however, when dealing with {\em unsafe formulas} (quite rare in real-world models), it still requires a disjunctive normal form to ensure correctness.
	Notably, in connection to ADFs, the experiments of~\cite{ASDO2024} show that \tsconj outperforms \kadf on large-scale BN and ADF models.
	
	\toolname builds on the strengths of both symbolic and logic-based approaches by integrating BDD-based reasoning with insights from the BN domain.
	Unlike \adfbdd, which applies BDDs directly to ADF semantics, \toolname leverages the BN-ADF correspondence to import advanced computation techniques from tools like \aeon~\cite{BBKPS2020} (dealing with the case of many input nodes) or \biolqm.
	
	\section{Preliminaries}
	\label{sec:preliminaries}
	
	In this section, we briefly recall the syntax and semantics of Abstract Dialectical Frameworks (ADFs).
	
	Let \(\threed = \{\fval, \tval, \uval\}\) be the 3-valued domain where \(\fval\), \(\tval\), and \(\uval\) denote the \emph{false}, \emph{true}, and \emph{unknown} values, respectively.
	Then \(\twod = \{\fval, \tval\}\) is the 2-valued (Boolean) domain.
	The logical connectives used in this paper include \(\land\) (\emph{and}), \(\lor\) (\emph{or}), \(\neg\) (\emph{negation}), \(\Rightarrow\) (\emph{implication}), and \(\Leftrightarrow\) (\emph{equivalence}). 
	
	\begin{definition}[\cite{BW2010}]
		An ADF is a tuple \(D := (S, C)\) where \(S\) is a \emph{finite} set of arguments and \(C := \{\acceptcon_s\}_{s \in S}\) consists of acceptance conditions (one for each argument in $S$), corresponding to propositional formulas \(\acceptcon_s := s' \in S \mid \fval \mid \tval \mid \neg \acceptcon \mid (\acceptcon \circ \acceptcon)\), such that $\circ$ can be one of the binary operators $\land, \lor, \Rightarrow, \Leftrightarrow$.
	\end{definition}
	
	The original definition of~\cite{BW2010} also assumes a link relation of dependencies between arguments to mirror the attack relation of AFs~\cite{Dung1995}. However, this relation can be inferred from the acceptance conditions and is therefore commonly omitted in recent literature.
	
	While the semantics of AFs is based on sets of arguments~\cite{Dung1995}, that of ADFs is based on \threeints~\cite{BSEWW2013}.
	A 3-valued \emph{interpretation} is a mapping \(I \colon S \to \threed\); it becomes a \twoint if \(I(s) \neq \uval\) for each \(s \in S\).
	An interpretation represents a belief assignment over the ADF arguments, with values mapping to ``argument holds'' ($1$), ``argument does not hold'' (0), and ``unknown'' ($\star$).
	
	An \emph{information ordering} \(\inorder\) is defined as the reflexive transitive closure of \(\inorderst\) where \(\inorderst\) is an ordering on \(\threed\) as follows: \(\uval \inorderst \fval\), \(\uval \inorderst \tval\), and there is no other relation.
	Intuitively, values $0$ and $1$ carry more information than $\uval$.
	
	The information ordering extends to \threeints by \(I_1 \inorder I_2\) \ifftext \(I_1(s) \inorder I_2(s)\) for each \(s \in S\), and \(I_1 \inorderst I_2\) \ifftext \(I_1 \inorder I_2\) and there is some \(s\) such that \(I_1(s) \inorderst I_2(s)\).
	Intuitively, $I_1 \inorder I_2$ if for every $s \in S$, either $I_1(s) = \uval$, or $I_1(s) = I_2(s)$. 
	That is, $I_2$ can only differ from $I_1$ in arguments that are unknown in $I_1$.
	
	Then, the \emph{partial evaluation} \(\varphi[I]\) of \(\varphi\) \wrt \(I\) is defined by \(\varphi[I] := \varphi[s/I(s) \mid s \in S,~I(s) \in \twod]\).
	That is, \(\varphi[I]\) is the propositional formula obtained from \(\varphi\) by substituting each argument $s$ that has a known Boolean value in $I$ with this Boolean value $I(s)$.
	Note that if $I$ is a \twoint, $\varphi[I]$ always simplifies to $1$ or $0$. We then write $\models \varphi$ to denote that for every \twoint $I$, $\varphi[I]$ simplifies to $1$ (in other words, $\varphi$ is a tautology on \twoints).
	
	Finally, the \emph{characteristic operator} \(\charoper(I) = I'\) uniquely defines the \threeint \(I'\) as follows: 
	\(I'(s) = \tval\) \ifftext \( \models \varphi_s[I]\) ($\varphi_s[I]$ is a 2-valued tautology), \(I'(s) = \fval\) \ifftext \(\models \neg\varphi_s[I]\) ($\varphi_s[I]$ is a 2-valued contradiction), and \(I'(s) = \uval\) otherwise.
	We are now able to define the semantics of ADFs:
	
	\begin{definition}[\cite{BW2010}]
		Let \(D = (S, C)\) be an ADF, and \(I\) a \threeint; then $I$ is:
		\begin{itemize}
			\item \emph{Admissible} \ifftext \(I \inorder \charoper(I)\). Intuitively, any belief held by such $I$ is fully justified by the acceptance conditions of $D$.
			\item \emph{Complete} \ifftext \(I = \charoper(I)\). Fixed point of $\charoper$; includes all additional beliefs that logically follow from $I$ itself.
			\item \emph{Grounded} \ifftext \(I\) is \(\inorder\)-minimal complete interpretation. The least committed point of view admitting the maximum number of $\star$ arguments.
			\item \emph{Preferred} \ifftext \(I\) is \(\inorder\)-maximal admissible interpretation. The most informative admissible interpretation with the least number of $\star$ arguments.
			\item \emph{\twomodel} \ifftext \(I\) is complete and 2-valued. In other words, a complete interpretation with no $\star$ arguments.
		\end{itemize}
	\end{definition}
	
	Finally, we recall the definition of a \stbmodel, which is inspired by the \emph{stable model semantics} in logic programming~\cite{GL1988}.
	
	\begin{definition}[\cite{BW2010}]\label{def:stabel-model}
		Let \(D = (S, C)\) be an ADF and \(I\) be an arbitrary but fixed 2-valued interpretation of $D$.
		First, we define the \emph{reduced ADF} \(D^{I} := (S^{I}, C^{I})\) where \(S^{I} := \{s \in S \mid I(s) = \tval\}\), \(C^{I} := \{ \varphi_{s}^{I} \}_{s \in S^{I}}\) and $\varphi_{s}^{I} = \varphi_s[s'/\fval \text{~for~} s' \in S \text{~s.t.~} I(s') = 0]$.
		\(I\) is a stable model of \(D\) \ifftext \(I\) is a \twomodel of \(D\) and for the \grint \(G\) of \(D^{I}\), we have that \(I(s) = \tval\) implies \(G(s) = \tval\). Intuitively, all held beliefs in such $I$ are well-founded (i.e., free of circular reasoning). Every argument that holds is justified by an acyclic chain of support.
	\end{definition}
	
	We denote by \(\twomset{D}\) and \(\stbset{D}\) the sets of 2-valued and \stbmodels of \(D\), respectively; by \(\adset{D}\), \(\coset{D}\), \(\grset{D}\), and \(\prset{D}\) we denote the sets of admissible, complete, grounded, and preferred interpretations of \(D\), respectively.
	It has been shown that \(\stbset{D} \subseteq \twomset{D} \subseteq \prset{D} \subseteq \coset{D} \subseteq \adset{D}\), as well as \(\grset{D} \subseteq \coset{D}\), and \(|\grset{D}| = 1\)~\cite{BSEWW2013}.
	For the \twomodel semantics, the enumeration of all \twomodels is generally intractable; its decision (resp.\ counting) version is NP-complete (resp.\ \#P-complete) in general~\cite{BSEWW2013}.
	Deciding whether an ADF has a \stbmodel is \(\Sigma^{\text{P}}_2\)-complete~\cite{BSEWW2013}, but it remains an open problem to precisely classify the enumeration complexity (or counting complexity) of \stbmodels in ADFs.
	This applies to the complete and preferred semantics too.
	Overall, the enumeration problem in ADFs is very challenging, in particular for large-scale problem instances.
	
	\begin{example}
		Consider the ADF \(D = (S, C)\) with \(S = \{a, b, c\}\) and \(\acceptcon_{a} = \tval, \acceptcon_{b} = \neg a \lor c, \acceptcon_{c} = b\).
		It has five \adints:
		\begin{align*}
			I_1 &= \{a \mapsto \tval, b \mapsto \fval, c \mapsto \fval\}, I_2 = \{a \mapsto \tval, b \mapsto \tval, c \mapsto \tval\}, \\
			I_3 &= \{a \mapsto \tval, b \mapsto \uval, c \mapsto \uval\}, \\
			I_4 &= \{a \mapsto \uval, b \mapsto \uval, c \mapsto \uval\}, I_5 = \{a \mapsto \uval, b \mapsto \tval, c \mapsto \tval\}.
		\end{align*}
		Here, \(I_1\), \(I_2\), and \(I_3\) are \coints, such that $I_1$ and $I_2$ are preferred, but also \twomodels.
		The \twomodel \(I_2\) is not stable because the reduced ADF \(D^{I_2}\) where \(D^{I_2} = D\) has the \grint \(\{a \mapsto \tval, b \mapsto \uval, c \mapsto \uval\}\), which does not satisfy the last condition of Definition~\ref{def:stabel-model}.
		Only \(I_1\) is a \stbmodel of \(D\).
		\label{example:adf}
	\end{example}
	
	\section{Analysis Methods}
	\label{sec:analysis-methods}
	
	In this section, we introduce the methods used by \toolname. The primary goal of \toolname in each case is to directly characterize the result set of the reasoning problem using a BDD, providing a compact representation of the full solution set while avoiding enumeration.
	Proofs are delegated to the \emph{Supplementary Material}~\cite{pastva2026bass}.
	
	\subsection{Symbolic representation of ADFs}
	
	A reduced, ordered BDD~\cite{DBLP:journals/tc/Bryant86} is a rooted, directed, acyclic graph with two terminal nodes ($\mathbf{0}$ and $\mathbf{1}$) that encodes a Boolean function $f: \mathbb{B}^k \to \mathbb{B}$. Each non-terminal $u$ (also called a \emph{decision node}) has two successors, a positive and a negative one, denoted $\mathit{pos}(u)$ and $\mathit{neg}(u)$. Each such $u$ also has an associated decision variable $\mathit{var}(u) \in \{ 1, \ldots, k \}$ corresponding to the inputs of $f$. On every path in this graph, each decision variable must appear at most once, and the variables respect a fixed linear order. Each $u$ then represents a function $f_u: \mathbb{B}^k \to \mathbb{B}$ via the Shannon expansion:
	\begin{equation*}
		f_u(x) = (x_{\mathit{var}(u)} \land f_{\mathit{pos}(u)}(x)) \lor (\neg x_{\mathit{var}(u)} \land f_{\mathit{neg}(u)}(x))
	\end{equation*}
	
	For the terminal nodes, we assume $f_{\mathbf{0}}(x) = 0$ and $f_{\mathbf{1}}(x) = 1$. A BDD $F$ is a symbolic encoding of a Boolean function $f$ if $f = f_{\mathit{root}(F)}$. Intuitively, to \emph{evaluate} $f$ using its BDD, we start at the root and, at each decision node $u$, pick the successor based on the value of $x_{\mathit{var}(u)}$ until we reach a terminal node, which determines the output. For additional details, we refer readers to~\cite{clarke2018handbook}, which provides a comprehensive overview of the data structure and its applications in symbolic graph analysis.
	
	Assuming a fixed variable ordering, a BDD representation is canonical (i.e., a specific $f$ is always represented by the same BDD). In the following, we thus largely consider BDDs interchangeable with Boolean functions, and unless stated otherwise, each Boolean function is represented by its BDD. However, in the worst case, the size of the BDD (the number of its nodes) can be exponential in $k$ (the number of function inputs). Given BDDs $A$ and $B$, we can compute the BDD of $A \circ B$ in $\mathcal{O}(|A| \cdot |B|)$, where $\circ$ is any binary Boolean operator and $|A|$ denotes the number of nodes in $A$. While BDDs can often represent complex Boolean functions compactly, keeping the BDD size as small as possible is of utmost importance due to this quadratic complexity.
	
	BDDs allow for the existential quantification of an input:
	\begin{equation*}
		\exists x_i .\, f(x) = f(x[i/1]) \lor f(x[i/0])
	\end{equation*}
	Here, $x[i/b]$ is a copy of $x$ with the $i$-th position set to $b$. Notice that $g(x) = \exists x_i .\, f(x)$ no longer depends on its $i$-th argument. Consequently, we can equivalently interpret $g$ as a function over the domain $\mathbb{B}^{k-1}$. As we will discuss shortly, this is critical when symbolically encoding Boolean relations, because existential quantification allows us to perform projections onto specific components of the relation.
	
	\paragraph{Symbolic ADF encoding} Assume ADF $D = (S, C)$ with an arbitrary but fixed ordering of arguments $S = \{ s_1, \ldots, s_n \}$. Then, we observe that a set $X \subseteq \twod^n$ of \twoints can be represented by its  characteristic function $f_X: \twod^n \to \twod$ s.t. $f_X(x) = 1$ iff $x \in X$. In that regard, every acceptance condition $\varphi_s: \twod^n \to \twod$ also trivially has a corresponding BDD representation. To represent \threeints, we define a \emph{dual encoding} using additional Boolean variables $(s^\top, s^\bot)$ defined for each $s \in S$, such that $(1,0) \equiv 1$, $(0,1) \equiv 0$, and $(1,1) \equiv \star$. Here, $(0,0)$ is an invalid combination which we exclude by requiring $s^\top \lor s^\bot$ in all symbolic representations of \threeints. Using this encoding, a set $Y \subseteq \threed^n$ can be represented as a function $f_Y: \mathbb{B}^{2n} \to \mathbb{B}$. Finally, an arbitrary relation $R$ between 2-valued and 3-valued interpretations can be encoded via a BDD using a characteristic function $f_R: \mathbb{B}^n \times \mathbb{B}^{2n} \to \mathbb{B}$ (or $f_R: \mathbb{B}^{3n} \to \mathbb{B}$). For the sake of readability, we will consider interpretations to be interchangeable with their Boolean encodings: when $I$ is an interpretation, we can write $f(I)$ to denote ``function $f$ evaluated on the Boolean encoding of $I$'', as long as $f$ is of the correct type (i.e., $f: \mathbb{B}^n \to \mathbb{B}$ for 2-valued, and $f: \mathbb{B}^{2n} \to \mathbb{B}$ for 3-valued interpretations).
	
	\paragraph{On multi-valued decision diagrams} Note that three-valued decision diagrams (DDs) could be also used to encode sets of 3-valued interpretations. Assuming $s^\top$ and $s^\bot$ follow each other in the variable ordering, a single three-valued node could replace up to three standard BDD nodes, resulting in a constant-factor DD size reduction. Such an approach is always a trade-off between DD size and internal overhead of the implementation, as considering additional levels complicates the decision node implementation. Here, we chose the BDD encoding mainly due to more mature and widespread support for BDD implementations compared to multi-valued DDs. Do note that this choice does not affect the asymptotic complexity of our algorithms in terms of the number of symbolic DD operations.
	
	\paragraph{Symbolic computation of characteristic operator} In addition to acceptance conditions and sets of interpretations, dual encoding allows us to directly represent the characteristic operator $\charoper{}$. Specifically, let $f: \twod^n \to \twod$ be a function over the arguments of $D$. Then, let us define:
	
	\begin{equation}
		\label{eq:dual-function}
		\begin{aligned}
			dual(f) &= \exists s_1, \ldots, s_n.~ \bigg(f(s_1, \ldots, s_n)~\land \\
			&\big(\bigwedge_{s_i \in S} (s_i \Rightarrow s_i^\top) \land (\neg s_i \Rightarrow s_i^{\bot})\big)\bigg)\\    
		\end{aligned}
	\end{equation}
	
	Here, note that while the whole expression uses $s_i^\top$, $s_i^\bot$ as well as $s_i$, due to the existential quantification, $dual(f)$ only depends on $s_i^\top$ and $s_i^\bot$. As such, we can treat $dual(f)$ as a function $dual(f): \twod^{2n} \to \twod$ which is satisfied by a 3-valued interpretation $I \in \threed^n$ \ifftext there exists a 2-valued interpretation $I' \in \twod^n$ s.t. $I \inorder I'$ and $f(I') = 1$. In other words, $dual(f)$ is satisfied by every \threeint that can be narrowed down to a \twoint satisfying $f$. With this in mind, let us define $(\varphi_s^\top, \varphi_s^\bot)$ as $(dual(\varphi_s), dual(\neg \varphi_s))$ for every $s \in S$. Now, let $I$ be a 3-valued interpretation and let $(a, b) =(\varphi_s^\top(I), \varphi_s^\bot(I))$. It then trivially follows that $\charoper(I) = 1$ iff $(a,b) = (1,0)$, $\charoper(I) = 0$ iff $(a,b) = (0,1)$, and $\charoper(I) = \star$ iff $(a, b) = (1,1)$. In other words, the pairs of functions $(\varphi_s^\top, \varphi_s^\bot)$ together define the dual encoding of $\charoper: \threed^n \to \threed^n$.
	
	\begin{example}
		Recall the ADF from Example~\ref{example:adf}, where $\varphi_a = 1$, $\varphi_b = \neg a \lor c$, and $\varphi_c = b$. The dual encoding of $\charoper{}$ is then defined by $(\varphi_a^\top, \varphi_a^\bot) = (1, 0)$, $(\varphi_b^\top, \varphi_b^\bot) = (a^\bot \lor c^\top, a^\top \land c^\bot)$, and $(\varphi_c^\top, \varphi_c^\bot) = (b^\top, b^\bot)$. Using this representation of $\charoper{}$, it is easy to verify that $I_5 = \{ a \mapsto \star, b \mapsto 1, c \mapsto 1 \}$ is not complete, but is admissible: $(\varphi_a^\top(I_5), \varphi_a^\bot(I_5)) = (1, 0) \equiv 1 \not\equiv \star$ (but $\star \inorder 1$), $(\varphi_b^\top(I_5), \varphi_b^\bot(I_5)) = (1, 0) \equiv 1$, and $(\varphi_c^\top(I_5), \varphi_c^\bot(I_5)) = (1, 0) \equiv 1$.
		\label{example:dual-encoding}
	\end{example}
	
	This dual encoding of acceptance conditions is widely used in BN tools (\aeon, \tsconj, \mpbn, and \biolqm all use it to some extent). However, in all of these tools, the translation from $f$ to $dual(f)$ is done \emph{syntactically}, typically by constructing the DNF representation of $f$ and then substituting positive literals for $s^\top$ and negative literals for $s^\bot$. Tools like \mpbn and \tsconj try to avoid this translation in certain special cases, but in the worst case default to DNF. This translation of $\varphi$ into DNF is often a bottleneck, as it can exponentially increase the size of the formula~\cite{TBPS2024}. While conceptually simple, we believe to be the first in this context to define this \emph{direct} propositional mapping (Equation~\ref{eq:dual-function}). This allows us to construct the BDDs of $\varphi^\top$ and $\varphi^\bot$ fully symbolically by transforming $\varphi$ instead of translating it into DNF. As an initial experiment, we implemented this singular optimization into the tool \aeon, which is included in the final evaluation.
	
	\subsection{Computation of 2-valued models}
	
	Given ADF $D = (S, C)$, the symbolic characterization of \twomodels (i.e. the set $\twomset{D}$) is straightforward. The following Boolean function is satisfied exactly by all \twomodels of $D$:
	\begin{equation}
		f_{\twomset{D}} = \bigwedge_{s_i \in S} s_i \Leftrightarrow \varphi_{s_i}
	\end{equation}
	
	This characterization is a well-known fact both in the ADF and BN communities and its correctness trivially follows from the observation that $\charoper{(I)}(s) = \varphi_s(I)$ when $I$ is a \twoint. Furthermore, the BDD of this formula can be constructed using just $\mathcal{O}(|S|)$ BDD operations. 
	
	\begin{algorithm}[t]
		\KwData{set of BDDs $Q$ (e.g. $Q = \{ a \Leftrightarrow \varphi_a \mid a \in S \}$ for $\twomset{D}$)}
		$R \gets \top$\;
		\While{$Q \not= \emptyset$}{
			$O \gets X \in Q \text{ s.t. BDD size } |X \land R| \text{ is minimal}$\;
			$R \gets R \land O$\;
			$Q \gets Q \setminus \{ O \}$\;
		}
		\textbf{return} $R$\;
		\caption{Greedy optimization algorithm for performing $n$-ary conjunction of BDDs.}\label{alg:greedy-merge}
	\end{algorithm}
	
	\paragraph{Conjunction optimization} Here, the intermediate results of applying the conjunction operator can be often much larger than the final BDD of $f_{\twomset{D}}$. To mitigate this problem, we utilize a simple greedy optimization scheme (Algorithm~\ref{alg:greedy-merge}), which increases the number of BDD operations to $\mathcal{O}(|S|^2)$, but is in practice often much faster because it can maintain dramatically smaller intermediate BDD sizes. Usage of such optimization schemes is fairly well known in the BDD community. However, we are not aware of other BDD-based tools in this area performing such optimization, thus we feel it should be highlighted. In the following, we use Algorithm~\ref{alg:greedy-merge} to construct the BDD of every similar $n$-ary conjunction operator.
	
	\subsection{Computation of admissible and complete interpretations}
	
	Using the dual encoding defined at the beginning of this section, the characterization of admissible and complete interpretations (i.e. the sets $\adset{D}$ and $\coset{D}$) is similar to that of \twomodels:
	
	\begin{align}
		f_{\adset{D}} =& \bigwedge_{s_i \in S} ((\varphi_{s_i}^\top \Rightarrow s_i^\top) \land (\varphi_{s_i}^\bot \Rightarrow s_i^\bot)) \\
		f_{\coset{D}} =& \bigwedge_{s_i \in S} ((\varphi_{s_i}^\top \Rightarrow s_i^\top) \land (\varphi_{s_i}^\bot \Rightarrow s_i^\bot) ~\land \\\nonumber
		&((s_i^\top \land s_i^\bot) \Rightarrow (\varphi_{s_i}^\top \land \varphi_{s_i}^\bot)))
	\end{align}
	
	\vspace{1pt}
	
	\begin{proposition}\label{prop:ad-co-char}
		The above BDD-based characterizations of $\adset{D}$ and $\coset{D}$ are correct, in the sense that the models of $f_{\adset{D}}$ (resp.\ $f_{\coset{D}}$) correspond exactly to the admissible (resp.\ complete) interpretations of the given ADF $D$.
	\end{proposition}
	
	\subsection{Computation of preferred interpretations}
	
	For preferred interpretations, the computation is more involved, as we need to select the $\inorder$-maximal interpretations from the set $\coset{D}$. A naive, often adopted approach, is to iteratively build a set $C$ of candidate interpretations using the following method: select an interpretation $I$ from $\coset{D}$, remove all $I' \in C$ and $I' \in \coset{D}$ such that $I' \inorderst I$, then insert $I$ into $C$. Upon termination, the set $C$ contains exactly the $\inorder$-maximal interpretations. However, this procedure requires anywhere from $|\prset{D}|$ to $|\coset{D}|$ iterations to converge, which can be intractable for large solution sets. Here, we instead present a symbolic algorithm which avoids this bottleneck, requiring only $|S|$ iterations.
	
	Our method is based on the following procedures, which are individually not new, but represent a somewhat less common usage of BDDs:
	
	\begin{itemize}
		\item Given a Boolean function $f: \twod^n \ \to \twod$ represented as a BDD $B$, it is possible to compute a satisfying valuation $x_1, \ldots, x_n$ (i.e. $f(x_1, \ldots, x_n) = 1$) with the least number of positive literals (i.e. $|\{ i \mid x_i = 1 \}|$ is minimal) in time $\mathcal{O}(|B|)$. We denote this operation $\mathit{leastVal}(f)$.
		\item Given Boolean functions $f_1, \ldots, f_m$, we can build the function $f_{\#k}$ satisfied \ifftext exactly $k$ functions from this list are satisfied, i.e. $f_{\#k}(x_1, \ldots, x_n) = (k = \sum_{i=1}^{m} f_i(x_1, \ldots, x_n))$. Note that when $f_i$ are literals (i.e. $f_i = x_i$), the size and construction time of such BDD are both known to be in $\mathcal{O}(m \cdot k)$. We denote this operation $\#_k(f_1, \ldots, f_m)$.
		\item Given a Boolean function $f: \twod^n \to \twod$, we can compute (in $\mathcal{O}(n)$ BDD operations) a function $f': \twod^n \to \twod$ such that $f'(y) = 1$ iff there exists some $x$ such that $f(x) = 1$ and $\{ i \mid x_i = 1 \} \subseteq \{ i \mid y_i = 1\}$. We denote this operation $\mathit{weakening}(f)$.
	\end{itemize}
	
	In the context of dual symbolic encoding of ADFs, we can assign a more intuitive meaning to these three operations: $\mathit{leastVal}(f)$ is $\inorder$-maximal interpretation out of all interpretations satisfying $f$ (specifically an interpretation with the least number of $\star$ values). Using $f_i = s_i^\top \land s_i^\bot$, the function $\#_k(f_1, \ldots, f_n)$ encodes exactly the \threeints with $k$ arguments set to $\star$. Finally, $\mathit{weakening}(f)$ encodes the set of interpretations that ``weaken'' the interpretations of $f$ by replacing an arbitrary number of fixed values with $\star$. In other words, if $f(I) = 1$, then $\mathit{weakening}(f)(I') = 1$ for every $I' \inorder I$. Using these operations, we can then build a \fulsym optimization procedure (see Algorithm~\ref{alg:preferred}).
	
	\begin{algorithm}[t]
		\KwData{BDD representing the set $\coset{D}$}
		$\prset{D} \gets \bot$\;
		\While{\upshape $\coset{D} \not= \bot$}{
			$I \gets \mathit{leastVal}(\coset{D})$\;
			$k \gets |\{ i \mid I(s_i^\top) = 1 \land I(s_i^\bot) = 1 \}|$\;
			$\mathit{prf} \gets \coset{D} \land \#_k(s_1^\top \land s_1^\bot, \ldots, s_n^\top \land s_n^\bot)$\;
			$\prset{D} \gets \prset{D} \lor \mathit{prf}$\;
			$\coset{D} \gets \coset{D} \land \neg \mathit{weakening}(\mathit{prf})$\;
		}
		\textbf{return} $\prset{D}$\;
		\caption{Iterative symbolic algorithm to compute the preferred interpretations $\prset{D}$ based on the known complete interpretations $\coset{D}$.}\label{alg:preferred}
	\end{algorithm}
	
	\begin{proposition}\label{prop:pref-char}
		Algorithm~\ref{alg:preferred} correctly computes the set of preferred interpretations $\prset{D}$ of an ADF $D$, based on its set of complete interpretations $\coset{D}$, using symbolic operations over BDDs.
	\end{proposition}
	
	\subsection{Computation of stable models}
	
	Finally, we discuss our \fulsym method for computing stable 2-valued models. This method follows the principles established in the computation of preferred models, but extends them to the stable 2-valued semantics. To start, we recall the following proposition from literature:
	
	\begin{definition}
		Consider an ADF \(D = (S, C)\).
		Given two interpretations \(I_1\) and \(I_2\) of \(D\), we say \(I_1 \leq_t I_2\) \ifftext \(I_1(s) = \tval\) implies \(I_2(s) = \tval\) for all \(s \in S\).
	\end{definition}
	
	\begin{proposition}[Proposition 3.8 of~\cite{Strass2013}]\label{theo:two-stb}
		Consider an ADF \(D = (S, C)\).
		Given two \stbmodels \(I_1\) and \(I_2\) of \(D\), it holds that \(I_1 \leq_t I_2\) implies \(I_1 = I_2\).
	\end{proposition}
	
	Intuitively, this proposition states that \stbmodels are a subset of 2-valued models minimal in terms of their true arguments (that is, w.r.t. $\leq_t$). It is not necessarily true that every such minimal \twomodel is a \stbmodel; we still have to account for the grounded interpretation of the reduced ADF. However, we can speed up the \stbmodel search significantly by pre-selecting \twomodels minimal w.r.t. $\leq_t$. This provides an advantage compared to tools such as \kadf, which computes \stbmodels by checking each \twomodel. Furthermore, to search for such minimal models, we can adapt the symbolic optimization method that we propose for preferred interpretations.
	
	To complete the procedure, we need to compute the grounded interpretation of the reduced ADF corresponding to each candidate \twomodel. To do this, we combine the direct and dual encoding of ADFs within one BDD, allowing us to represent a relation between \twomodels and their corresponding \threeints in the reduced ADF. Using this combined representation we can compute grounded models $G$ corresponding to \twomodels $I$ and only select those combinations where $I(s) = 1$ implies $G(s) = 1$. As the final step, we project our symbolic representation back to the direct encoding, obtaining the final \stbmodels. The complete procedure is presented in Algorithm~\ref{alg:stable}.
	
	\begin{algorithm}[t]
		\SetKwRepeat{Fix}{do}{until fixed-point}
		\KwData{BDD representing the set $\twomset{D}$}
		$\stbset{D} \gets \bot$\;
		\While{\upshape $\twomset{D} \not= \bot$}{
			$I \gets \mathit{leastVal}(\twomset{D})$\;
			$k \gets |\{ i \mid I(s_i) = 1 \}|$\;
			$\mathit{min} \gets \twomset{D} \land \#_k(s_1, \ldots, s_n)$\;
			$\stbset{D} \gets \stbset{D} \lor \mathit{min}$\;
			$\twomset{D} \gets \twomset{D} \land \neg \mathit{weakening}(\mathit{min})$\;
		}
		$\stbset{D} \gets \stbset{D} \land \bigwedge_{s \in S} (s \Rightarrow (s^\top \land s^\bot)) \land (\neg s \Rightarrow (\neg s^\top \land s^\bot))$\;
		\Fix{}{
			\For{$s \in S$}{
				$\mathit{setToOne} \gets \stbset{D} \land (s^\top \land s^\bot) \land (\varphi_s^\top \land \neg \varphi_s^\bot)$\;
				$\stbset{D} \gets (\stbset{D} \land \neg \mathit{setToOne}) \lor \mathit{setToOne}[s^\bot / \neg s^\bot]$\;
			}    
		}
		\textbf{return} $\exists s_1^\top, s_1^\bot, \ldots, s_n^\top, s_n^\bot.~(\stbset{D} \land \bigwedge_{s \in S} (s \Rightarrow (s^\top \land \neg s^\bot)))$\;
		\caption{Iterative symbolic algorithm to compute the \stbmodels $\stbset{D}$ based on the known \twomodels $\twomset{D}$.}\label{alg:stable}
	\end{algorithm}
	
	\begin{proposition}\label{prop:stable-char}
		Algorithm~\ref{alg:stable} correctly computes the set of stable models $\stbset{D}$ of an ADF $D$, based on its set of 2-valued models $\twomset{D}$ and the grounded interpretations of their corresponding reduced ADFs.
	\end{proposition}
	
	\subsection{Implementation notes}
	
	\paragraph{Free input arguments} An ADF can contain a \emph{free input} argument $s$ where $\varphi_s = s$ (called an \emph{input variable} in BNs). This is especially common in ADFs based on Boolean networks, where such arguments represent free external stimuli. For such $s$, we can simplify the optimization process used in Algorithms~\ref{alg:preferred} and~\ref{alg:stable}. In Algorithm~\ref{alg:preferred}, it is clear that no preferred interpretation can have such $s$ set to $\star$ (only $0$ or $1$ are possible). Similarly, in Algorithm~\ref{alg:stable}, it is clear that no \stbmodel can have such $s$ set to $1$ (only $0$ is possible). When computing $\prset{D}$ and $\stbset{D}$, our implementation therefore restricts $\coset{D}$ and $\twomset{D}$ with these additional constraints. These directly propagate into the computation of $\coset{D}$ and $\twomset{D}$, simplifying the set of solutions. As such, \toolname can sometimes compute $\prset{D}$ or $\stbset{D}$ even if it fails to compute $\coset{D}$ or $\twomset{D}$.
	
	\paragraph{Implementation} The implementation of \toolname originally started within the \aeon toolbox, which we previously enhanced with the \fulsym BDD transformation of acceptance conditions into dual encoding and greedy conjunction optimization. However, this is the first publication where this enhancement is presented and tested. Compared to this initial implementation, \toolname also offers the presented \fulsym computation of preferred and \stbmodels, which are not present in \aeon. Finally, \toolname also adopts a novel BDD library called \ruddy which implements CPU cache-friendly optimizations of the BDD data structure from~\cite{PH2023}. Compared to \aeon, this makes \toolname a fully-featured ADF solver. \toolname is available on GitHub and is written entirely in Rust.
	
	\section{Evaluation}
	\label{sec:evaluation}
	
	\begin{table*}
		\centering
		\renewcommand{\arraystretch}{1.1}
		\setlength{\tabcolsep}{3pt}
		\begin{tabular}{c|c|c|c|c|c||c|c|c|c|c|c}
			& \toolname & \kadf & \godiamond & \yadf & \adfbdd & \aeon & \tsconj & \fasp & \mpbn & \biolqm & \he \\\hline         
			admiss. & \yesmark & \yesmark & \yesmark & \yesmark & \nomark & \nomark & \nomark & \nomark & \nomark & \yesmark & \nomark \\
			complete & \yesmark & \yesmark & \yesmark & \yesmark & \yesmark & \yesmark & \nomark & \nomark & \nomark & \yesmark & \nomark \\
			preferred & \yesmark & \yesmark & \yesmark & \yesmark & \nomark & \yesmark & \yesmark & \nomark & \yesmark & \yesmark & \nomark \\
			2-val. & \yesmark & \yesmark & \yesmark & \nomark & \yesmark & \yesmark & \yesmark & \yesmark & \yesmark & \yesmark & \yesmark \\
			stable & \yesmark & \yesmark & \yesmark & \yesmark & \yesmark & \nomark & \nomark & \nomark & \nomark & \nomark & \nomark \\
		\end{tabular}
		\vspace{5pt}
		\caption{Overview of the tested tools and ADF problems that they can solve out-of-the-box. Tools on the right-hand-side require translation into the \texttt{.bnet} BN format.}
		\label{tab:tool-problem-matrix}
	\end{table*}

	\begin{table*}
		\centering    
		\renewcommand{\arraystretch}{1.1}
		\setlength{\tabcolsep}{2pt}    
		\begin{tabular}{c|C{34pt}|C{36pt}|C{43pt}|C{32pt}|C{43pt}||C{32pt}|C{45pt}|C{32pt}|C{32pt}|C{40pt}|C{35pt}}
			& \toolname & \kadf & \texttt{goDiam.} & \yadf & \adfbdd & \aeon & \tsconj & \fasp & \mpbn & \biolqm & \he \\\hline    
			
			admiss. & 
			\cellcolor{darkgreen}\smallcell{\textbf{695 (338) 991s}} & 
			\cellcolor{lightgreen}\smallcell{298 (0) \hspace{15pt} 1796s} & 
			\cellcolor{lightgreen}\smallcell{170 (0) \hspace{15pt} 2063s} & 
			\cellcolor{lightgreen}\smallcell{201 (0) 1992s} & 
			\cellcolor{lightgrey}\smallcell{N/A} & 
			\cellcolor{lightgrey}\smallcell{N/A} & 
			\cellcolor{lightgrey}\smallcell{N/A} & 
			\cellcolor{lightgrey}\smallcell{N/A} & 
			\cellcolor{lightgrey}\smallcell{N/A} & 
			\cellcolor{lightgreen}\smallcell{292 (0) \hspace{15pt} 1807s} &
			\cellcolor{lightgrey}\smallcell{N/A}
			\\\hline
			
			comp. & 
			\cellcolor{darkgreen}\smallcell{\textbf{974 (6{\tiny+79}) 404s}} & 
			\cellcolor{darkgreen}\smallcell{\textbf{883 (23) 602s}} & 
			\cellcolor{lightgreen}\smallcell{190 (0) \hspace{15pt} 2023s} & 
			\cellcolor{darkgreen}\smallcell{\textbf{535 (21) 1320s}} & 
			\cellcolor{lightgreen}\smallcell{476 (0) \hspace{10pt} 1430s} & 
			\cellcolor{lightgreen}\smallcell{963 (0) 432s} & 
			\cellcolor{lightgrey}\smallcell{N/A} & 
			\cellcolor{lightgrey}\smallcell{N/A} & 
			\cellcolor{lightgrey}\smallcell{N/A} & 
			\cellcolor{lightgreen}\smallcell{619 (0) \hspace{15pt} 1136s} &
			\cellcolor{lightgrey}\smallcell{N/A}
			\\\hline
			
			pref. & 
			\cellcolor{darkgreen}\smallcell{\textbf{982 (24{\tiny+15}) 389s}} & 
			\cellcolor{lightgreen}\smallcell{1081  (0) \hspace{10pt} 181s} & 
			\cellcolor{lightgreen}\smallcell{215  (0) \hspace{15pt} 1958s} & 
			\cellcolor{lightgreen}\smallcell{342  (0) 1732s} & 
			\cellcolor{lightgrey}\smallcell{N/A} & 
			\cellcolor{lightgreen}\smallcell{900  (0) 562s} & 
			\cellcolor{darkgreen}\smallcell{\textbf{1103 (1) \hspace{10pt} 134s}} & 
			\cellcolor{lightgrey}\smallcell{N/A} & 
			\cellcolor{lightgreen}\smallcell{1065  (0) 219s} & 
			\cellcolor{lightgreen}\smallcell{652 (0) \hspace{15pt} 1067s} &
			\cellcolor{lightgrey}\smallcell{N/A}
			\\\hline
			
			2-val. & 
			\cellcolor{darkgreen}\smallcell{\textbf{1142 (3) 52s}} & 
			\cellcolor{lightgreen}\smallcell{1116  (0) \hspace{10pt} 102s} & 
			\cellcolor{lightgreen}\smallcell{1130  (0) \hspace{20pt} 74s} & 
			\cellcolor{lightgrey}\smallcell{N/A} & 
			\cellcolor{lightgreen}\smallcell{950 (0) \hspace{10pt} 456s} & 
			\cellcolor{lightgreen}\smallcell{1139  (0) 59s} & 
			\cellcolor{lightgreen}\smallcell{1104  (0) \hspace{10pt} 134s} & 
			\cellcolor{lightgreen}\smallcell{1123  (0) \hspace{10pt} 94s} & 
			\cellcolor{lightgreen}\smallcell{1079 (0) \hspace{10pt} 189s} & 
			\cellcolor{lightgreen}\smallcell{1038 (0) \hspace{10pt} 272s} &
			\cellcolor{lightgreen}\smallcell{1120 (0) \hspace{10pt} 98s}
			\\\hline
			
			stable & 
			\cellcolor{darkgreen}\smallcell{\textbf{1149 (50) 36s}} & 
			\cellcolor{darkgreen}\smallcell{\textbf{1115 (16) 105s}} & 
			\cellcolor{lightgreen}\smallcell{202 (0) \hspace{15pt} 2002s} & 
			\cellcolor{lightgreen}\smallcell{65 (0) 2269s} & 
			\cellcolor{lightgreen}\smallcell{571 (0) \hspace{10pt} 1249s} & 
			\cellcolor{lightgrey}\smallcell{N/A} & 
			\cellcolor{lightgrey}\smallcell{N/A} & 
			\cellcolor{lightgrey}\smallcell{N/A} & 
			\cellcolor{lightgrey}\smallcell{N/A} & 
			\cellcolor{lightgrey}\smallcell{N/A} &
			\cellcolor{lightgrey}\smallcell{N/A}
			\\\hline
		\end{tabular}
		\vspace{5pt}
		\caption{Summary table of benchmark results. Each cell lists (a) the total number of completed benchmarks; (b) in parentheses, the number of benchmarks \emph{uniquely} solved by that tool; (c) the PAR2 runtime score (in seconds). For \toolname, we also list (using $+x$) the number of instances solved uniquely by \toolname and \aeon, but no other tool. Bold text highlights tools that uniquely solved at least one problem instance.}
		\label{tab:results-summary}
	\end{table*}
	
	\paragraph{Tools} To evaluate the performance of \toolname, we compare its runtime against the current fastest ADF solvers: \kadf~\cite{LMNWW2022} (version 2021-03-31), \godiamond~\cite{SE2017} (version 2.0.2), \yadf~\cite{BDHLW2020} (version 2.2) and \adfbdd~\cite{EGRW2022} (version 0.3.1). Furthermore, we test the following state-of-the-art BN analysis tools that can be used to solve certain ADF problems: \aeon~\cite{BBHPSS2022} (version 1.2.5), \tsconj~\cite{TBPS2024} (version 2024-10-29), \fasp~\cite{TBS2023CP} (version 0.3.0), \mpbn~\cite{VML2024} (version 0.4.1), \biolqm~\cite{Naldi2018} (version 0.6.1, using BDDs) and \he (version 1.1)~\cite{HSBMBT2025}.
	Since these tools can typically only solve a subset of ADF problems, we provide an overview of the tested combinations in Table~\ref{tab:tool-problem-matrix}.
	All BDD-based tools order the variables as presented in the input file and we do not perform any specific tuning in this regard.
	
	\paragraph{Test instances} Our benchmark dataset consists of $1{,}237$ ADF instances.
	Out of these, $245$ are Boolean networks from the \emph{Biodivine Boolean Models} (BBM) dataset~\cite{PSBBH2023} converted into the ADF format. The rest are benchmarks previously used in~\cite{BDHLW2020,LMNWW2022,EGRW2022}.
	Specifically, this includes (a)~test instances based on the benchmarks from ICCMA 2017 ($601$ instances; from the \yadf website\footnote{https://www.dbai.tuwien.ac.at/proj/adf/yadf/}); (b)~test instances based on the ICCMA 2019 competition ($276$ instances; from the \kadf repository\footnote{https://bitbucket.org/andreasniskanen/k-adf/}); (c)~ADFs of metro networks and grids used in~\cite{BDHLW2020} ($115$ instances; from the \yadf website).
	We measured up to $1{,}076$ arguments per benchmark instance ($135$ average; $80$ median) and up to $923{,}346$ links (i.e., dependencies between arguments)  per benchmark instance ($16{,}470$ average; $280$ median).
	
	\paragraph{BN conversion} For the BN tools, the ADF instances also had to be converted into the \texttt{.bnet} format.
	Here, we encountered a compatibility issue, as \texttt{.bnet} does not support the XOR operator, which is prominently used in some of the ADF benchmarks.
	We attempted to automatically rewrite these ADF instances without using XOR, but in $72$ cases, this resulted in extremely large BNs that we had to exclude from the final analysis.
	To make the presented conclusions fair to all tools, we decided to completely exclude these $72$ benchmarks from the final comparison, leaving $1{,}165$ instances.
	However, we did test ADF tools (including \toolname) on these inputs, and $52{/}72$ instances failed across all problems and tools.
	As such, ADF-specific tools are better at handling large XOR problem instances, but the advantage is not very significant on the benchmarks we have available.
	
	\begin{figure*}
		\centering
		\begin{minipage}{0.33\linewidth}
			\centering
			\includegraphics[width=1.0\linewidth]{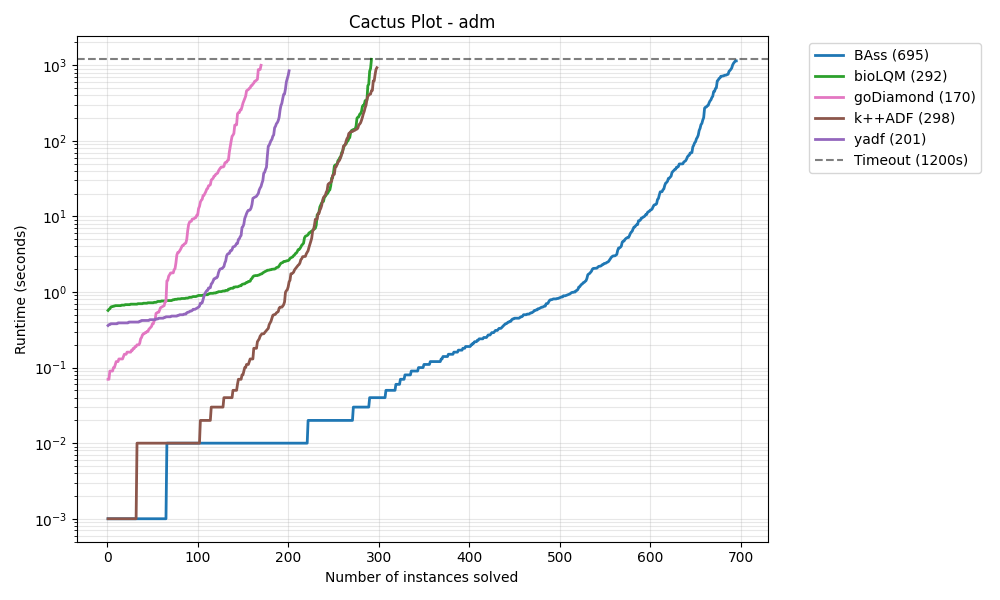}
			
			(a) Admissible interpretations
		\end{minipage}
		\begin{minipage}{0.33\linewidth}
			\centering
			\includegraphics[width=1.0\linewidth]{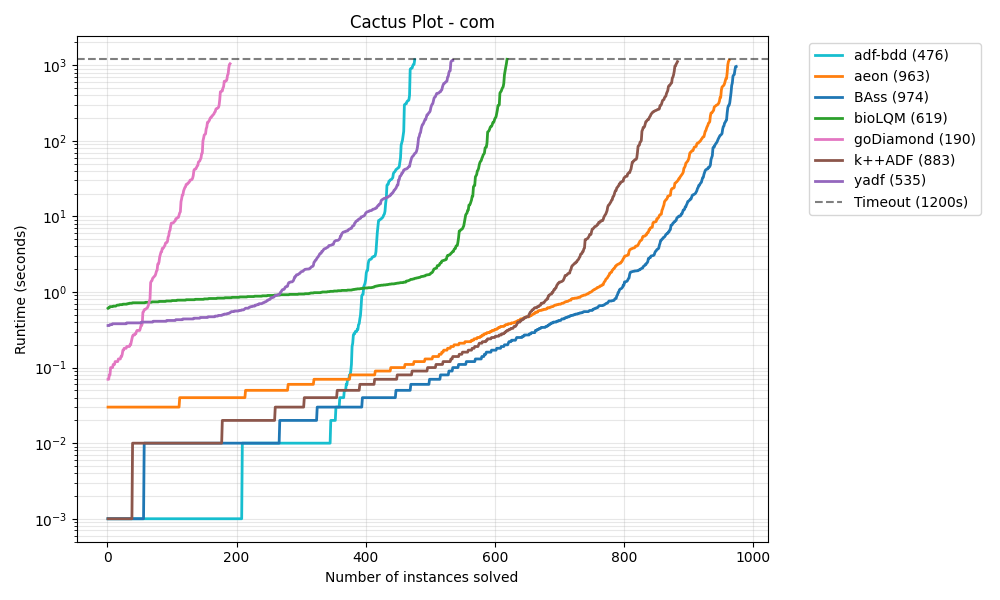}
			
			(b) Complete interpretations
		\end{minipage}
		\begin{minipage}{0.32\linewidth}
			\centering
			\includegraphics[width=1.0\linewidth]{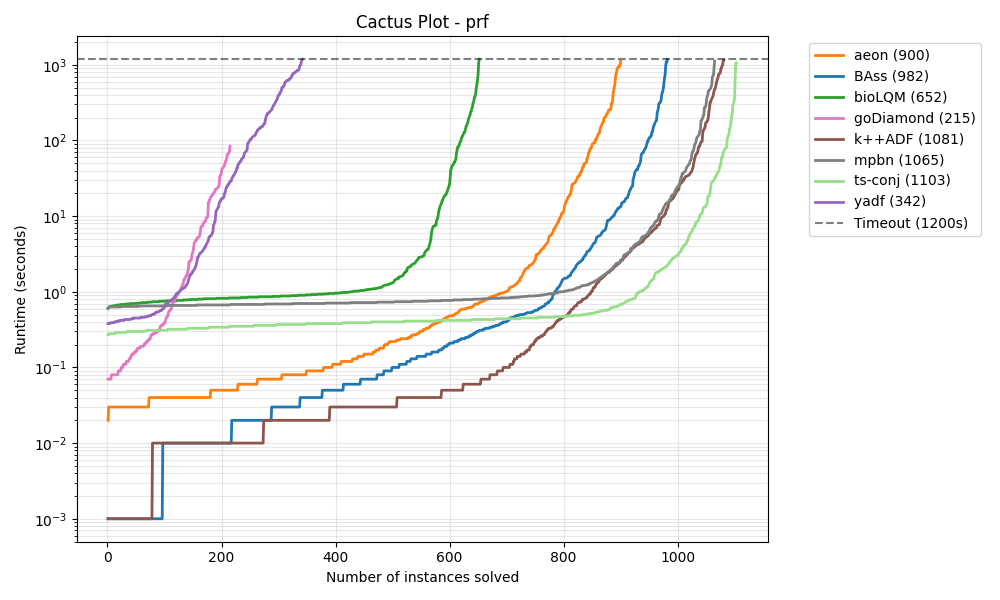}
			
			(c) Preferred interpretations
		\end{minipage}
		
		\vspace{4pt}
		
		\begin{minipage}{0.32\linewidth}
			\centering
			\includegraphics[width=1.0\linewidth]{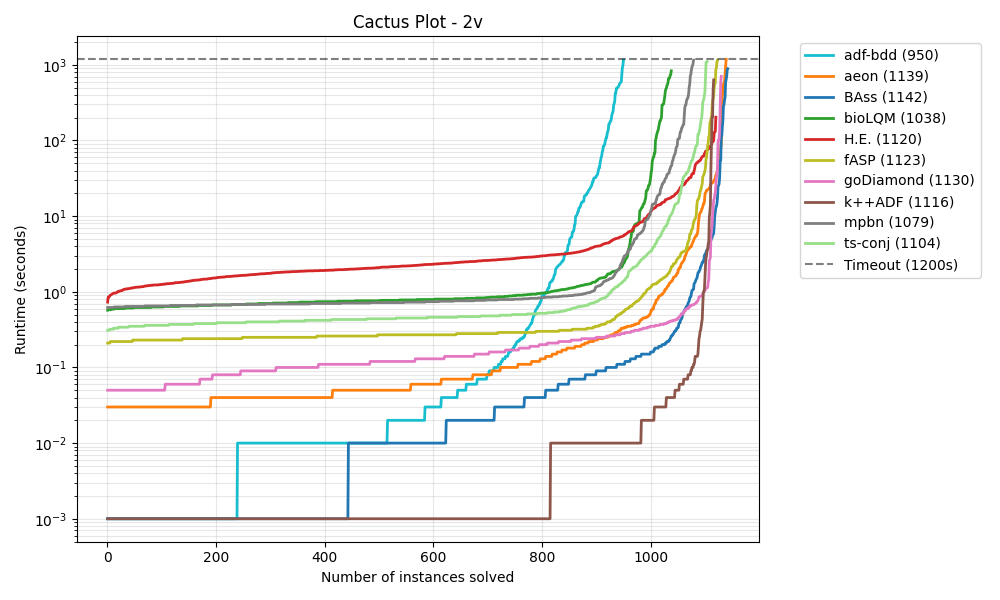}
			
			(d) 2-valued models
		\end{minipage}
		\begin{minipage}{0.32\linewidth}
			\centering
			\includegraphics[width=1.0\linewidth]{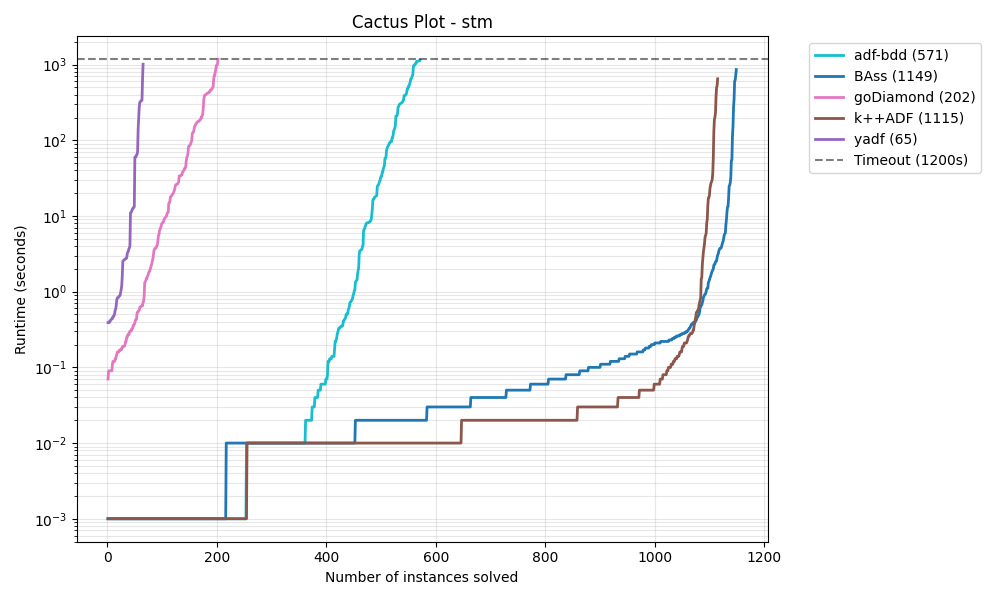}
			
			(e) Stable models
		\end{minipage}
		
		\caption{Cactus plots showing the runtime of individual tools across different ADF problem categories.}
		\label{fig:performance-plots}
	\end{figure*}
	
	\begin{table*}[t]
		\centering    
		\renewcommand{\arraystretch}{1.1}
		\setlength{\tabcolsep}{4pt}    
		\begin{tabular}{c|c|l|r|r}
			BBM ID & $|S|$ & Source & \#Preferred & \#Complete \\\midrule
			222 & 206 & \cite{bbm-222} & 687194767360 & 6699707176948710000 \\
			116 & 144 & \cite{bbm-116} & 423188831391449000 & 9389013197839268000000000000 \\
			223 & 141 & \cite{bbm-223} & 408333504 & 57879370298664 \\
			087 & 133 & \cite{bbm-087} & 17179869184 & 16680058317523500 \\
			118 & 121 & \cite{bbm-116} & 67545198297874400 & 540713282220605900000000000 \\
			112 & 112 & \cite{bbm-116} & 135266304 & 7926623204940 
		\end{tabular}
		\caption{Largest BBM Boolean networks in terms of arguments (i.e., $|S|$) where \toolname{} was uniquely able to identify the whole solution space of complete and preferred interpretations.}
		\label{tab:biological-results}
	\end{table*}

	\paragraph{Methodology} We ran all experiments on an AMD EPYC 7713 64-core server with 1TB of RAM, such that up to $48$ experiments were executed in parallel.
	For each experiment, we set a time limit of $1{,}200$ seconds. We did observe out-of-memory errors and related segmentation faults for certain tools (namely \yadf, \kadf, \he, and \biolqm).
	For simplicity, our main results count these as failures (the same way as timeouts), but they are logged separately in the raw dataset available in the attached reproducibility artifact.
	To compare the overall runtime, we use the PAR2 score, which is commonly used in solver competitions.
	The PAR2 score is the average runtime across all benchmarks, while counting timeouts or errors as $2 \cdot 1200$ seconds. 
	To mitigate effects of disk I/O and printing operations, the benchmarks only count the solutions (i.e. results are not printed or written to disk). 
	This ensures comparable experimental environment for all tools.
	As a separate experiment, we compared the full output of all tools for correctness using benchmarks where all tools finished in $\leq10s$, ensuring tools agree in their output.
	
	\paragraph{Results discussion} The measured results are presented in summary Table~\ref{tab:results-summary} and Figure~\ref{fig:performance-plots}.
	First, let us compare the performance of BDD-based tools.
	Here, \biolqm and \adfbdd represent the most straightforward implementations, but do not achieve good performance compared to SAT/ASP-based tools.
	Meanwhile, \aeon and \toolname are often comparable to SAT/ASP.
	In particular, note that \toolname solves unique benchmarks across all categories.
	Nevertheless, for complete, preferred, and stable interpretations, SAT/ASP still holds unique advantage in some problem instances (here, the two approaches appear somewhat complementary).
	Also notice that \toolname solved the most instances (and had the best PAR2 score) in the admissible, complete and stable model categories. 
	For preferred interpretations, \toolname is limited by its reliance on the full knowledge of $\coset{D}$.
	Let us also compare \toolname to \aeon (which represents an experimental implementation of a subset of the presented methods).
	In the complete and 2-valued categories, \toolname outperforms \aeon slightly due to its more efficient BDD implementation, but otherwise relies on the same methodology.
	In the preferred category, we see a clear improvement due to the new fully symbolic procedure.
	Finally, we also list the number of benchmarks uniquely solved by \toolname and \aeon, but no other tool, as this collectively represents the contribution of methods presented in this paper compared to the remaining tools.
	
	\subsection{Biological applicability}
	\label{sec:case-study}
	
	We introduced methods that reveal large sets of preferred interpretations (minimal trap spaces) and complete interpretations (percolated trap spaces). Table~\ref{tab:biological-results} showcases the six largest models that can only be handled by \toolname. 
	
	Having these solution sets available provides a significant benefit to the systems biology community. Complete knowledge of minimal trap spaces is crucial when determining the landscape of biological phenotypes in a Boolean model~\cite{HAH2015,KHNS2020} and detecting its \emph{phenotype determining nodes}. Meanwhile, percolated trap spaces (and to some extent all trap spaces) are relevant for BN control, as each trap space matching a desired biological profile can be used to derive a control strategy that drives the network towards this phenotype~\cite{TKPR2025,TP2024,FTS2020}. Finally, percolated trap spaces are used to form \emph{succession diagrams}~\cite{rozum2021parity}, which represent the gradual irreversible biological commitments of the model. 
	
	Even when the complete solution set is not required, or it is impossible to analyze each solution individually, having the solution set represented as a BDD provides a unique advantage: Due to its structure, one can \emph{uniformly} sample solutions from a BDD. Meanwhile, sampling solutions using a solver is typically biased by the underlying algorithm. As such, BDDs allow statistical analysis on the solution space not available to solver-based methods.
	
	\section{Conclusion}
	\label{sec:conclusion}
	
	We have presented \toolname, a new BDD-based reasoning tool for Abstract Dialectical Frameworks that leverages recent theoretical advances connecting ADFs with BNs.
	\toolname builds on existing theoretical foundations of \aeon, \biolqm and \adfbdd, but provides novel symbolic algorithms and optimization strategies for computing admissible, complete, and preferred interpretations as well as 2-valued and \stbmodels.
	It is, to the best of our knowledge, the first BDD-based ADF solver to support \fulsym reasoning on the preferred and stable semantics.
	Our extensive evaluation over $1{,}200$+ ADF and BN instances demonstrates that \toolname dramatically outperforms existing BDD-based tools and matches or exceeds state-of-the-art SAT/ASP-based solvers, particularly for models involving large numbers of solutions.
	In such cases, the BDD-based representation can replace full solution enumeration, allowing, for example, provably uniform random sampling of solutions.
	This property is crucial to maintain statistical interpretability in biological applications.
	Our results highlight the promise of symbolic methods as a robust and scalable foundation for ADF reasoning.
	Notably, we have successfully applied \toolname to analyze previously unsolved biological networks.
	
	Future work will improve performance by exploring dynamic variable reordering, hybrid symbolic-SAT strategies (e.g., eliminate the need for full computation of $\coset{D}$ when computing $\prset{D}$), and BN reduction~\cite{VelizCuba2011,TP2024}. We also plan to expand support to additional semantics such as the naive, stage, and semi-stable semantics~\cite{GRS2021,ZVV2021}.
	
	\section*{Availability} 
	\toolname is available at \url{https://github.com/sybila/biodivine-bass}, with a benchmark reproducibility artifact at \url{https://doi.org/10.5281/zenodo.17794231}. 
	The appendices are published on arXiv~\cite{pastva2026bass}.
	
	\section*{Acknowledgments} 
	We acknowledge Ho Chi Minh City University of Technology (HCMUT), VNU-HCM for supporting this study.
	
	\section*{AI Declaration} 
	Generative AI (Cursor) was used during the development and testing of auxiliary features of \toolname; all core algorithms were written and tested manually. Generative AI (Gemini) was also utilized for proofreading the main text. However, no published content was directly generated by AI.


%% The file kr.bst is a bibliography style file for BibTeX 0.99c
\bibliographystyle{unsrt}
\bibliography{references}

\appendix

\section{Detailed proofs}

\begin{proposition}
	The above BDD-based characterizations of $\adset{D}$ and $\coset{D}$ are correct, in the sense that the models of $f_{\adset{D}}$ (resp.\ $f_{\coset{D}}$) correspond exactly to the admissible (resp.\ complete) interpretations of the given ADF $D$.
\end{proposition}
\begin{proof}
	The correctness of these characterizations is easily shown once we recall the correspondence between $(\varphi_s^\top, \varphi_s^\bot)$ and the dual encoding of $\charoper{}$.
	
	For $\adset{D}$, the condition states that $\charoper{(I)}(s) = 1$ implies $I(s) \in \{1, \star\}$, $\charoper{(I)}(s) = 0$ implies $I(s) \in \{0, \star\}$, and $\charoper{(I)} = \star$ implies $I(s) = \star$. Overall, this ensures $I \inorder \charoper{(I)}$.
	
	For $\coset{D}$, we reuse the $\adset{D}$ conditions (since $\coset{D} \subseteq \adset{D}$), and introduce a constraint stating that $I(s) = \star$ implies $\charoper{(I)}(s) = \star$. Together with the conditions of $\adset{D}$, this guarantees $I = \charoper{(I)}$.
\end{proof}

\begin{proposition}
	Algorithm~2 correctly computes the set of preferred interpretations $\prset{D}$ of an ADF $D$, based on its set of complete interpretations $\coset{D}$, using symbolic operations over BDDs.
\end{proposition}
\begin{proof}
	The procedure effectively works by identifying preferred interpretations with a fixed number of $\star$ values, starting with the least amount of $\star$ values. Interpretations with the same amount of $\star$ values are always incomparable and can be thus considered ``concurrently'' within one iteration.
	
	Therefore, in every iteration, the procedure selects one remaining complete interpretation with the least amount of $\star$ values.
	At this point, not only is such interpretation guaranteed to be $\inorder$-maximal (i.e. preferred), but all other interpretations $\mathit{prf}$ with the same number of $\star$ values are preferred as well.
	Consequently, these are ``moved'' to the symbolic set $\prset{D}$.
	Finally, we compute the set of all $I'$ s.t. $I' \inorder I$ for some $I$ from $\mathit{prf}$ and remove these from $\coset{D}$.
	This ensures that $\mathit{leastVal}(\coset{D})$ is again a preferred interpretation (or $\coset{D}$ is empty).
	
	Note that if the ADF has any \twomodels, these are always preferred and always all detected in the first iteration of Algorithm~2.
	Subsequently, in each iteration, we completely remove one ``size'' of 3-valued preferred interpretations, meaning the algorithm converges in $\mathcal{O}(|S|)$ iterations.
\end{proof}

\begin{proposition}[Proposition 3.8 of (Strass 2013)]
	Consider an ADF \(D = (S, C)\).
	Given two \stbmodels \(I_1\) and \(I_2\) of \(D\), it holds that \(I_1 \leq_t I_2\) implies \(I_1 = I_2\).
\end{proposition}

\begin{proposition}
	Algorithm~3 correctly computes the set of stable models $\stbset{D}$ of an ADF $D$, based on its set of two-valued models $\twomset{D}$ and the grounded interpretations of the corresponding reduced ADFs.
\end{proposition}
\begin{proof}
	Lines 2-8 of Algorithm~3 filter out the $\leq_t$-minimal \twomodels, such that the process directly mirrors the one from Algorithm~2.
	Subsequently, Line 9 sets the dual encoding $(s^\top, s^\bot)$ based on the value of each argument in the current set of candidates.
	If $s$ is true, the dual encoding of $s$ is set to $\star$ (i.e. $(s^\top, s^\bot) = (1,1)$).
	If $s$ is false, it is set to $0$ (i.e. $(s^\top, s^\bot) = (0, 1)$).
	This turns the candidate set into a symbolically represented relation, associating each \twomodel with a \threeint of its reduced ADF (recall that in the reduced ADF, arguments are fixed to $0$ if they are $0$ in the \twomodel).
	
	On Lines 10-16, the algorithm performs grounding on the whole set of candidate interpretations at once.
	
	Line 12 selects all candidates where a argument $s$ is~$\star$ (i.e. $(s^\top, s^\bot) = (1, 1)$), but where $\charoper{(I)}(s) = 1$ (i.e. $(\varphi^\top_s, \varphi^\bot_s) = (1, 0)$).
	
	On Line 13, these candidates are then removed from the candidate set and replaced with ones where the dual encoding of $s$ is set to $1$ (i.e. $(s^\top, s^\bot) = (1, 0)$) (note that we only need to substitute $s^\bot$, since $s^\top$ is already guaranteed to be $1$).
	
	Once the grounding process is complete, Line 16 performs two important steps: First, an added constraint ensures that we only keep those \twomodels $I$ and their grounded counterparts $G$ where $I(s) = 1$ implies $G(s) = 1$. Subsequently, existential quantification projects the whole relation back to a set of \twomodels, which at this point are guaranteed to be stable.
\end{proof}

\end{document}